\documentclass[journal]{IEEEtran}
\pdfoutput=1

\ifCLASSINFOpdf
\else
\fi

\usepackage{graphicx}

\newtheorem{theorem}{Theorem}[section]

\newtheorem{proposition}[theorem]{Proposition}
\newtheorem{assumption}[theorem]{Assumption}

\newtheorem{definition}[theorem]{Definition}

\usepackage{cite}
\usepackage{booktabs}
\usepackage{amsmath}
\usepackage{amssymb}

\begin{document}

%
\title{smartSDH: An Experimental Study of Mechanism-Based Building Control}
%
%
%

\author{Ioannis C.\ Konstantakopoulos,
        Kristy A.\ Hamilton, Yashaswini Murthy,
        Tanya Veeravalli, Costas Spanos,~\IEEEmembership{Fellow,~IEEE,} and Roy Dong 
\thanks{Ioannis C.\ Konstantakopoulos is with Amazon.com; Kristy A.\ Hamilton is with  Department of Communication at University of California, Santa Barbara, CA, USA; Yashaswini Murthy, Tanya Veeravalli and Roy Dong are with the Department of Electrical and Computer Engineering, University of Illinois, Urbana-Champaign, IL, USA; Costas Spanos is with the Department of Electrical Engineering and Computer Sciences, University of California, Berkeley, CA, USA.}
}

\maketitle

\begin{abstract}
As Internet of Things (IoT) technologies are increasingly being deployed, situations frequently arise where multiple stakeholders must reconcile preferences to control a shared resource. We perform a 5-month long experiment dubbed `smartSDH' (carried out in 27 employees' office space) where users report their preferences for the brightness of overhead lighting. smartSDH implements a modified Vickrey-Clarke-Groves (VCG) mechanism; assuming users are rational, it incentivizes truthful reporting, implements the socially desirable outcome, and compensates participants to ensure higher payoffs under smartSDH when compared with the default outside option (i.e., the option chosen in the absence of such a mechanism). smartSDH assesses the feasibility of the VCG mechanism in the context of smart building control and evaluated smartSDH's effect using metrics such as light level satisfaction, incentive satisfaction, and energy consumption. Although previous studies on the theoretical aspects of the mechanism indicate user satisfaction, our experiments indicate quite the contrary. We found that the participants were significantly less satisfied with light brightness and incentives determined by the VCG mechanism over time. These data suggest the need for more realistic behavioral models to design IoT technologies and highlights difficulties in estimating preferences from observable external factors such as atmospheric conditions.
\end{abstract}


%
\IEEEpeerreviewmaketitle

\section{Introduction}
\label{sec:Intro}
%
%
%
%
\IEEEPARstart{R}{ecent} advances in sensing, actuation, and communication technologies have allowed an unprecedented level of control over the behavior of our devices and an unprecedented fidelity of information about the state of our systems. This has found application in a wide variety of `smart' decision making processes, including the management of parking spaces, the monitoring of water usage, and the energy-efficient operation of instrumented homes and offices \cite{shah2016survey}.

However, as the set of application domains for these `smart' algorithms grow, we find ourselves in scenarios where multiple stakeholders care about the state of a shared resource. For example, consider the control of lights in an office space. At home, one can easily adjust the light settings in a non-problematic way, as they are dictator of their own home's illumination. In contrast, for an office setting where multiple co-workers have different preferences about the intensity of lights, the definition of what `should' happen is not as clear. Put another way, as we start to look at applications with a shared resource, \emph{we need to find ways to reconcile the opinions of multiple users.}

Building on our previous example, we can think of a mechanism that asks each person: ``How bright do you want the lights?'' Suppose there are two office occupants and our mechanism simply implements the average of the two votes. Furthermore, suppose your co-worker voted for 100\% intensity, and you wanted the lights at 75\% intensity. In this situation, you would vote for for 50\% intensity, as the mechanism would implement the average, which would be your desired lighting level. 
This simple thought experiment reveals a crucial point: humans have an incentive to \emph{strategically} report data to push mechanisms to their selfishly desired outcomes.

This problem is often studied in economics under the title of `mechanism design'. The Vickrey-Clarke-Groves (VCG) mechanism~\cite{vickrey:1961aa,groves:1973aa} is one of the most celebrated achievements in this area.For more information about VCG mechanisms, we refer the reader to the wonderfully accessible~\cite{Nisan2007}. VCG mechanisms have many desirable properties. First, VCG mechanisms implement the most socially beneficial outcome. Second, all participants are incentivized to truthfully reveal their preferences. Third, the mechanism can be designed such that every participant is better off participating in the mechanism, either due to the outcome chosen or due to some endogenously determined payment.

The most common application of VCG mechanisms is in auction-like settings. In these settings, people submit bids to `win' a good. This is \emph{not} the case here, and we briefly elucidate some of the finer nuances of our application at a higher level. 

The VCG mechanism can be applied to any setting with a set of mutually exclusive outcomes experienced by a set of agents. In auctions, the different, mutually exclusive outcomes are which bidder wins the good (there can be a single winner amongst the bidders). In our application, the outcomes are the 3 possible light settings. In its full generality, the VCG mechanism simply takes into account user preferences and outputs a socially-optimal outcome after having reconciled these reported preferences. In smartSDH, users bid for an optimal light setting. The chosen light setting may not be the optimal light setting for every user, yet, due to the shared nature of the office space, only one light setting can be chosen (i.e. the 3 outcomes are mutually exclusive) and every user experiences the effect of the outcome. Depending on the socially-optimal outcome, the mechanism then issues payments to users based on the preferences of all other users, incentivising truthful reporting of their preferences.   

While the VCG mechanism has many desirable properties when all participants act as rational agents, the literature of real-world implementations outside the context of auctions is limited, as we will discuss further in Section~\ref{sec:back}. Therefore, the purpose of this experiment is to 1) assess the feasibility of the modified VCG mechanism in the context of smart building control in group settings and 2) evaluate the effect of smartSDH on light level satisfaction, incentive satisfaction, and energy consumption across time. These main research objectives are delineated further in Section \ref{sec:model}.

Furthermore, due to the application domain, minor modifications were required to implement VCG mechanisms in a real-time setting where users cannot walk away and the outside option is more complicated than `non-participation'; more details can be found in Section~\ref{sec:model}.

The rest of the paper is organized as follows. In Section~\ref{sec:back}, we discuss the existing literature most relevant to our work, and contextualize our contribution. In Section~\ref{sec:model}, we present the mathematical formalism for modeling the decision making processes of users and the formulation of the mechanism implemented by smartSDH. In Section~\ref{sec:method}, we outline our experimental setup. In Section~\ref{sec:results}, we present the statistical analysis of our experimental data. In Section~\ref{sec:discuss}, we discuss and interpret the results, as well as provide closing remarks summarizing our discoveries and potential avenues for future work.

\section{Background and Related Work}
\label{sec:back}

In this section, we discuss the existing literature and works most closely related to ours. Broadly speaking, this literature can be broken down into two categories: works on building control and works on mechanism design.

\subsection{Building control}

Nearly 50\% of the energy consumed in the U.S. is accounted for by residential and commercial buildings~\cite{mcquade2009}, and well-designed algorithms for the control of lighting and heating-ventilation-air-conditioning systems in buildings promise significant benefits for the stability and efficiency of our power grid. 
Much of the research in building control focuses on finding algorithms that take into account the uncertainty about external factors and operating conditions, while satisfying the constraints introduced by user preferences~\cite{maAnderson2011ACC,Aswani:2012kx}.

Recent work in smart building infrastructures incorporate occupant preferences about thermal comfort~\cite{karmann2018percentage}, satisfaction/well-being~\cite{frontczak2012quantitative}, lighting comfort~\cite{baron1992effects}, acoustical quality~\cite{ryherd2008implications}, indoor air quality~\cite{sundell2011ventilation}, indoor environmental monitoring~\cite{jin2018automated}, privacy~\cite{jia2016privacy}, and productivity~\cite{wargocki2008indoor}, while simultaneously optimizing for energy efficiency and agile connectivity with the grid. Alongside these, ontological modelling has emerged as one of the latest avenues pursued in the domain of resource allocation through smart technologies~\cite{Nezami2019ontology,Seydoux2016ontology}. 

The most relevant work closely related to the experimental study presented in this paper concerns the theoretical study of implementing the VCG mechanism for the purpose of moderating the users' thermal comfort across multiple zones within a building \cite{Gupta2018Article}. The primary approach of \cite{Gupta2018Article} involves modelling of the users' comfort as a function and incentivising the users to aid in truthful representation of their preferences. Although their simulations and theoretical properties support the utility of such an approach in real life, it's lacking in strong experimental validation. We aim to bridge this gap by presenting a thorough experimental study of such a mechanism in real world. Consequently, our results indicate the need to resort to other mechanisms as the user satisfaction over time in reality does not align with the theoretically predictions of the VCG mechanism. 

Prior work has attempted to correct the information imbalance by modeling occupant preferences by the means of a multi-agent system~\cite{boman:1998aa}. An initial control policy is created from both the building manager's and the occupants' preferences. Then, a rule engine finds a compromise in the system, and iteratively performs a compromise. 
In another related work, ~\cite{clearwater1995auction} demonstrates an auction-based apparatus where each room in an office building makes a bid between the desired temperature and the actual temperature.

With regards to other building control scenarios, ~\cite{Ploennigs2017FromSM} carry out an experiment with sensors placed in buildings and use machine learning and statistical learning techniques with the information from the sensors to diagnose building operation problems related to energy usage and occupant comfort. 
Along a similar vein, ~\cite{egedorf2017adverse} develop a methodology using causal learning statistics to model occupants' thermal discomfort in smart buildings using temperature sensors. A more classical approach to understanding HVAC settings in a building using optimal control is explored in ~\cite{wang2008supervisory}. Here, a variety of optimal control and optimization methods are analyzed in the context of energy efficiency and/or cost-efficient control of HVAC systems, many of which utilize variations of different mechanisms.

A common goal for those implementing building control is energy-efficiency. ~\cite{chen2009design} propose a smart software system that balances the requests from different stakeholders (owners, operators, etc.) and also considers the preferences and locations of individual occupants. The framework enables building control while reducing energy consumption and maximizing occupant comfort.

In contrast to the above methods for human-building interaction, we present an experimental study of a different approach. At the core of our approach is the design of a mechanism that can find the shared conditions (e.g. lighting and HVAC) that fairly account for all occupants' preferences, and provides rewards to those who are willing to compromise on the shared conditions for their co-workers. Additionally, our mechanism can pass on some of the building manager's energy-efficiency incentives onto the building occupants, so they can also experience the benefit of energy-efficiency programs. Our framework is centered around a real-time application of a modified VCG mechanism; occupants are asked for their preferences on the shared source, and this information is used to calculate the {social optimum} and allocate rewards. The rewards are calculated such that occupants can only benefit from the mechanism, and none of the occupants have any incentive to misrepresent their preferences.

\subsection{Mechanism design}
\label{sec:mec_des}

Game-theoretic models have been widely used to model selfish agent behavior in a wide variety of applications, ranging from traffic network flow allocation to smart grid optimization. There is one thing common in the core of these problems: each user is trying to maximize a personal utility function (which can be function of various internal and external factors) subject to personal constraints. Mechanism design is an approach to strategically design incentives for these users towards an objective. 
Most game-theoretic analyses rely on the assumption that the utility functions are known a priori, which may not be a realistic assumption to make in real-world situations; this is especially true in the case of the energy industry where there are many `noise' variables affecting how well a building can operate. However there have been some attempts at user preference modelling in intelligent environments~\cite{Juan2019User}, some delving into argumentation methods in order to formalise commonsense reasoning~\cite{Teze2020User}.

While VCG mechanisms have many nice theoretical properties, it is rare to see their supposed effect in the real-world integrated with decision-making. There are some existing works in the body of literature that contribute to the simulation and/or deployment of VCG mechanisms \cite{Gupta2018Article}. For instance, the mechanism is applied to numerical examples of wholesale energy markets to achieve social optimality by incentivizing truth-telling~\cite{xu2017efficient}. Similarly, there exists compelling work on simulating the VCG mechanism results in optimal energy load scheduling where social welfare is also measured monetarily~\cite{samadi2011optimal}. With regards to human computation, crowdsourcing tasks and designing optimal pricing policies also require incentive-compatible mechanism design; strong theoretical guarantees are shown with regards to case studies done using Mechanical Turk~\cite{goel14mechanismdesign}.

We would like to emphasize that to the best of our knowledge, our work is the first to use VCG mechanisms in a real-world building control application. It is also one of the first papers to implement VCG mechanisms in real-world settings outside of the classical domains of auctions. We hope that this literature develops more in the coming years, as new technologies will have to reconcile the opinions of many users, and an understanding of how humans interact with different mechanisms is crucial for these systems to operate as desired.

\section{Model}
\label{sec:model}

In this section, we briefly introduce the mathematical formalism for modeling the decision making processes of users, as well as the formulation of our implemented mechanism. To recap, every user submits his/her bid for an optimal light setting, and the mechanism computes the socially acceptable light intensity although it need not be optimal for every individual user. In later sections, we'll use data from our experiments to explore the faithfulness of these models.

\subsection{User model}

First, we formally outline the assumptions of our human models. We begin by introducing some notation. Let \(\mathcal{I}\) denote the finite set of participants, and let \(\mathcal{X}\) denote the finite set of outcomes. Recall that our mechanism will choose one of the outcomes in $\mathcal{X}$ to implement. Let $n = |\mathcal{I}|$ denote the number of participants, and $m = |\mathcal{X}|$ denote the number of outcomes. Without any loss of generality, we'll let the labels of users be $\mathcal{I} = \{1,2,\dots,n\}$.

In our smartSDH experiment, the number of outcomes $m = 3$, corresponds to three different light settings: (`Normal', `Bright', `Very Bright') and we have $n = 27$ participants.

In order to impart a mathematical structure to our model, it is necessary to define certain notions.
\begin{definition}[Type of the user]
For each user $i \in \mathcal{I}$ and outcome $x \in \mathcal{X}$, let $\lambda_x^i \in \mathbb{R}$ denote the cost user $i$ incurs when the chosen outcome is $x$. 
We refer to the vector $\lambda^i = (\lambda_x^i)_{x \in \mathcal{X}} \in \mathbb{R}^m$ as the type of the user.
\end{definition}

Our mechanism will choose an outcome $x$ and issue each user $i \in \mathcal{I}$ a payment $p_i$. 
We wish to model people's preferences not only across outcomes $x$, but across outcome-incentive bundles $(x,p)$. This allows us to compensate users for compromising on the outcome; for example, if we choose an outcome $x$ that is particularly distasteful to a user $i \in \mathcal{I}$ (i.e. $\lambda_x^i$ is very large), then we can compensate by giving them a very large gift card (i.e. choosing $p_i$ to be very large as well). This concept is better illustrated as the utility function of the user. 

\begin{definition}[Utility Function]
Each user \(i \in \mathcal{I}\) makes decisions to maximize their utility function, which is parameterized by their type $\lambda^i$ and takes as inputs an outcome-incentive bundle $(x,p)$:
\begin{equation}
u(x, p; \lambda^i) = p - \lambda_x^i  
\label{eq:userform}
\end{equation}
\end{definition}

The interpretation is that when a user's type is $\lambda^i$, the chosen outcome is $x$, and the incentive payment received is $p$, their utility is $u(x,p;\lambda^i)$.

We are now ready to introduce our assumptions on our human decision-making model.

\begin{assumption}[User model]
\label{ass:user_model}
The user's type $\lambda^i$ captures all the relevant information for their decision making process and the utility function is given by~\ref{eq:userform}. 
\end{assumption}

Additionally, we will use the notation \(\lambda = (\lambda^1,\dots,\lambda^n)\) to denote the types of \emph{all} users. As per common game-theory convention, we will use the notation \(\lambda^{-i} = (\lambda^1,\dots,\lambda^{i-1},\lambda^{i+1},\dots,\lambda^n)\) to denote the types of all users other than user \(i\). Hence \(f(\lambda) = f(\lambda^i, \lambda^{-i})\). 

\begin{assumption}[Informational structure]
\label{ass:info_structure}

We assume that each user $i \in \mathcal{I}$ knows their own type $\lambda^i$. Additionally, no one other than $i$ knows $\lambda^i$. This includes other users $j \neq i$ and the mechanism designer.

\end{assumption}

Assumption~\ref{ass:info_structure} sits at the core of what we wish to model. When asked about their preferences, we assume users will answer based on their desired outcome, not their true preferences. Similar to the motivating example in Section~\ref{sec:Intro}, where the user suggests a wrong preference in an averaging mechanism, to obtain a final result of his preference, he might report a wrong type with skewed costs as well to bias the VCG output in his favour. Our goal is to design a mechanism that can calculate the outcome that is most socially desirable for all participants, \emph{without} access to information about the user types $\lambda$. In particular, our mechanism will ask users to report their types; we let their reported values be denoted $\widehat \lambda$. Our mechanism must decide payments and the outcome based only on these reported values $\widehat \lambda$.

\subsection{Vickrey-Clarke-Groves mechanisms}

Given the user models, we can analyze whether mechanisms will achieve desired outcomes. In this subsection, we'll define VCG mechanisms, outline the desirable properties these mechanisms have in theory, and discuss extensions required for our application.

First, let us define the socially desirable outcome.

\begin{definition}[Social welfare] The social welfare of an outcome \(x \in \mathcal{X}\) for a set of users of type \(\lambda = (\lambda^1,\dots,\lambda^n)\) is given by:
\begin{equation}
  s(x,\lambda) = -\sum_{i \in \mathcal{I}} \lambda_x^i  
\end{equation}
We define the social optima as the maximizers of the social welfare:
\begin{equation}
s^*(\lambda) = \arg\max_{x \in \mathcal{X}} s(x,\lambda)
\end{equation}

\end{definition}

Intuitively, the social welfare function calculates the sum total of everyone's utilities, excluding the incentive payments. In other words, the social welfare function assesses how good each outcome $x \in \mathcal{X}$ is for all participants without payment compensation. 

Now, we'll provide a definition of a class of mechanisms known as Vickrey-Clarke-Groves mechanisms, and outline some of their theoretically desirable properties as well as nuances in our experiment's particular implementation.

\begin{definition}[Vickrey-Clarke-Groves mechanisms]
Given the reported types \(\widehat \lambda\), a \emph{Vickrey-Clarke-Groves mechanism} chooses the outcome that maximises the social welfare:
\begin{equation}
  f(\widehat \lambda) = \arg\max_{x \in \mathcal{X}} s(x,\widehat \lambda)  
\end{equation}
Additionally, for each user $i \in \mathcal{I}$, it pays:
\begin{equation}
 p_i(\widehat \lambda) =
h_i(\widehat \lambda^{-i}) - \sum_{j \neq i} \widehat \lambda_{f(\widehat \lambda)}^j   
\end{equation}
Here, \(h_{i}\) is an arbitrary function which does not depend on the reported value $\widehat \lambda^i$ of user $i$. The form of \(h_{i}\) determines the individual optimality of the chosen output with respect to an output which is obtained in the absence of such a mechanism. Depending on the modelling of utilities in various scenarios and constraints over the types of users, \(h_{i}\) can be wisely chosen to ensure the optimality of the VCG mechanism.    

\end{definition} 

VCG mechanisms choose the outcome that is socially optimal for the reported types. Furthermore, it issues payments to every user \(i\) based on the utility of all other users from the chosen outcome. Intuitively, this means that once the payments are incorporated, the utility function of every user essentially becomes the social welfare function along with payment compensation. 
Formally, if we plug in $p_i$ to the utility function of user $i$, we get:
\begin{equation}
u(f(\widehat \lambda), p_i(\widehat \lambda); \lambda^i) = h_i(\widehat \lambda^{-i}) - \sum_{i \in \mathcal{I}} \widehat \lambda_{f(\widehat \lambda)}^i = h_i(\widehat \lambda^{-i}) + s(x, \lambda)   
\end{equation}
We next introduce the notion of incentive compatibility, in order to incentivise truthful reporting of preferences. This is an important notion in the implementation of the VCG mechanism, as it sets the rationale for the users to not lie about their types in order to maximise their utility. 
\begin{proposition}[Incentive compatibility]
\label{prop:ic}
For every set of user types \(\lambda\), and for every user \(i\) and reported type \(\widehat \lambda^i\), we have:
\begin{equation}
\label{eq:inc_comp}
u(f(\lambda), p_i(\lambda);\lambda^i) \geq
u(f(\widehat \lambda^i, \lambda^{-i}), p_i(\widehat \lambda^i, \lambda^{-i});\lambda^i)
\end{equation}

\end{proposition} 
The proof for Proposition~\ref{prop:ic} can be found in~\cite{Nisan2007}.

Incentive compatibility should be interpreted as follows: \(\lambda\) denotes the true types of all the users, and \(\hat{\lambda^{i}}\) is a possible lie which user \(i\) could report. The left-hand side of Equation~\ref{eq:inc_comp} is the utility user \(i\) gains when everyone is truthful. The right-hand side is the utility \(i\) derives when everyone except \(i\) truthfully responds. Equation \ref{eq:inc_comp} states that any such deviations will only lower user \(i\)'s utility. Thus, truthful reporting by all users forms a Nash equilibrium. Note that this utility function is still parameterized by the true type $\lambda^i$, as this determines the utility the user actually experiences.

Next, we cover another desirable property for VCG mechanisms: individual rationality. 
In the typical applications of the VCG mechanism, users can opt-out and choose an outside option. For example, in auctions, a participant can simply choose to not bid at all and walk away. They will receive no utility from receiving a good, and they will lose no utility from paying the auctioneer.

However, we note that in our smartSDH experiment, this is not as well-defined. Even if a participant does not report a vote to our mechanism, they must sit in their office space and experience the light chosen by our mechanism. This warrants a minor modification of the typical individual rationality constraint, since the outside option, becomes type-dependent, which we outline here. 
To understand this deviation from original VCG mechanism, we introduce the notion of nominal outcome which is analogous to the "outside option".

\begin{assumption}[Nominal outcome] 
\label{ass:nom}
Without any mechanism in place, there is a \emph{nominal outcome} \(x_0 \in \mathcal{X}\) that would occur. For each user, their \emph{outside option} is \(u(x_0,0;\lambda^i)\), which is the utility they get from the nominal outcome and no awarded points. It is important to note that $\lambda^i_{x_0}$ need not necessarily be equal to zero for all users.The modified VCG mechanism yields an output which is individually optimal with respect to the nominal outcome, if it satisfies ex-post individual rationality. 
A VCG mechanism has \emph{ex-post individual rationality} if for all user types \(\lambda\) and any user \(i\):
\begin{equation}
\label{eq:ex_post_ir}
u(f(\lambda),p_i(\lambda);\lambda^i) \geq
u(x_0,0; \lambda^i)
\end{equation}
\end{assumption}

The interpretation is that the nominal outcome (outside option) would be the chosen outcome in the absence of our mechanism. When we enforce ex-post individual rationality, this means that every user is better off when the mechanism is implemented, as opposed to when the mechanism does not exist at all. 

It's important to note that this is a deviation from typical applications of the VCG mechanism. In smartSDH, a user cannot `walk away' from their shared office space; rather, we have to ensure that their utility increases as a result of our mechanism's deployment. 
From a technical perspective, the important distinction is that the utility of the outside option depends on the user's type; this is what necessitates modification at a theoretical level.

\begin{assumption}[Bounded user utilities]
\label{ass:bounded}
The types of all users are bounded between \(0\) and \(\lambda_{max}\), i.e. \(\lambda_x^i \in [0,\lambda_{max}]\) for all \(i \in \) and \(x \in \mathcal{X}\).

\end{assumption}

\begin{proposition}[Existence of an ex-post individually rational mechanism]
\label{prop:ex_post_existence}
Using \(h_i(\widehat \lambda^{-i}) = n\lambda_{max}\), the VCG mechanism is ex-post individually rational.

\end{proposition}

\textit{Proof:} First, note that $n\lambda_{\max} - \sum_{i \in \mathcal{I}} \lambda_x^i \geq 0$ for any $x$ and any $\lambda$. Thus, $u(f(\lambda),p_i(\lambda);\lambda^i) \geq 0$. Next, note that since $\lambda_{x_0}^i \geq 0$ for all $i$ by assumption, we have $u(x_0,0; \lambda^i) \leq 0$. The inequality in Equation~\ref{eq:ex_post_ir} follows.

In the smartSDH experiment, $h_i(\widehat \lambda^i) = n\lambda_{max}$ happens to be independent of the types of the users. It can be thought of as the basic wage that every user is given for participating in the study.

Lastly, we outline the last modification needed to implement a VCG mechanism in our office space setting. Users can come and go at any point in time, and they can modify their reported preferences at any point in time as well. The modification needed from the classical VCG mechanism applications is the temporal aspect of this. 

Essentially, this modification is quite simple: transform all the quantities discussed above into rates. Rather than interpreting $\lambda_x^i$ as the cost user $i$ incurs when the outcome is $x$, we interpret $\lambda_x^i$ as the cost per hour user $i$ incurs. Similarly, if the VCG mechanism decides to pay user $i$ a payment of $p_i$, they are paid $p_i$ points per hour. Then, any time a user enters, leaves, or modifies their vote, we treat that segment as one round of the mechanism, weighted by its duration.

More formally, suppose the set of users logged on and their reported preferences $\widehat \lambda$ are constant on the time interval $[t_0,t_1]$. Then, in that time interval, our mechanism chooses the outcome $f(\widehat \lambda)$. Each user $i$ who is logged on in that time interval receives a reward $(t_1 - t_0) p_i(\widehat \lambda)$ for their participation during the time interval $[t_0,t_1]$. 

\subsection{Contributions and goal of the investigation}

The goal of the present investigation was to test the application of the modified VCG mechanism to a smart building control application, dubbed ``smartSDH'' and to experimentally verify its utility. As mentioned previously, a VCG mechanism selects the socially optimal outcome among a set of possible outcomes -- in this case the preferred brightness of overhead office lights -- and then issues payments to users based on the decisions of all other users. In this sense, it is a mechanism where being truthful is the best strategy for the individual and the group. Although VCG mechanisms have many desirable theoretical properties when all participants act as rational agents, real-world application of VCG mechanisms are limited.  

This study analyzes perceptions and behaviors of users in a shared office space who interacted with a VCG-operated smartSDH over a 5 month period. Participants interacted with smartSDH via a web portal, where they reported their preference for the brightness of the overhead office lights in real time so the VCG mechanism could determine the socially optimal light setting for the group. Users were allotted points according to the modified VCG mechanism. Whenever the total points earned crossed a threshold, a lottery was held for gift cards; over the 5-month period, we rewarded a total of \$2,900 worth in gift cards. Additionally, to emphasize the non-competitive aspect of sharing an office space, whenever the total points earned crossed another threshold, we hosted a catered lunch for all the participants of smartSDH. Because adaptation to technology often evolves over time~\cite{desanctis1994}, sometimes gradually and sometimes sporadically, we chose to evaluate the behavioral outcomes over the course of three distinct time periods (T1 = Wk1 - Wk7; T2 = Wk 8 - Wk15; T3 = Wk16 - Wk 22). Upon observation of data, the temporal changes were significant across periods of 8 weeks, hence the analysis breaks down the time period of data analysis to 8 weeks with no gaps in between. In such a setting, three research questions become salient:

\begin{itemize}

\item
\textbf{RQ1.} What is the influence of using the VCG-operated smartSDH on light and incentive satisfaction across time? How does an average user's lighting preference vary across time and in response to the incentives?

\item
\textbf{RQ2.} What is the influence of using the VCG-operated smartSDH on energy saving?

\item
\textbf{RQ3.} Is it possible to determine user's preferences from their immediate environment to facilitate smart lighting without user intervention? What is the relationship between light level preferences and atmospheric conditions including humidity, temperature, pressure, and solar radiation?

\end{itemize}

\section{Method}
\label{sec:method}

\subsection{Participants}

Twenty-seven undergraduate students from a public research university in California, USA participated in the longitudinal study in exchange for prize items, which were allotted based on their performance in the point based system. The preliminary survey was optional, and out of the 27 participants, 18 responded. The survey asked the students about what they considered to be their ideal incentives (gift cards vs fitbits, etc), the duration of their commute, modes of transportation they resorted to, time of arrival to the office, etc. The majority of users were men (gender: 73\% men, 22\% women, 5\% prefer not to respond) with 50\% of users between the ages of 22-25, while the rest were above 26 years old. The majority of users had incomes between \$35,000 and \$40,000, and the rest had lower incomes. Before commencing the experiment, users were somewhat satisfied with the light conditions in the office \((M = 4.39\), \(SD = 1.14)\) on a 5-point scale. As for ideal incentives, nearly everyone preferred gift cards from Amazon, iTunes, and Google Play (94\%) over the other options, which were complimentary vouchers for drinks at a campus coffeeshop, or lotteries for big-ticket items such as Apple Watches, Fitbits, and EarPods.

\subsection{Procedure}
Participants in the longitudinal study interacted with smartSDH in an open office space with cubicles on a university campus over a period of five months. The desks in the office space were divided into different lighting zones with a set of overhead lights serving as the primary source of illumination for each zone. The smartSDH operated independently for each zone. From the work hours of 9 AM to 5 PM local time, participants were only able to control their office lights through the smartSDH web portal. After work hours, the light switches returned to normal operation. 

Each participant had access to the password-protected smartSDH web portal. In addition to allowing users to vote for their light settings in real time, the web portal also provided participants with visualizations of the state of the office space. Users could view their personal point totals, the currently implemented light setting, and the global progress to the individual and communal incentive thresholds. They could also see which occupants were present in their zone. Users were able to monitor their floor's lighting lumen level and temperature in real-time, with a refresh interval of one second. A view of the designed portal can be seen in \ref{fig:portal}.
\begin{figure}[h]
    \centering
     \label{fig:portal}
    \includegraphics[width=0.47\textwidth, height=0.3\textwidth]{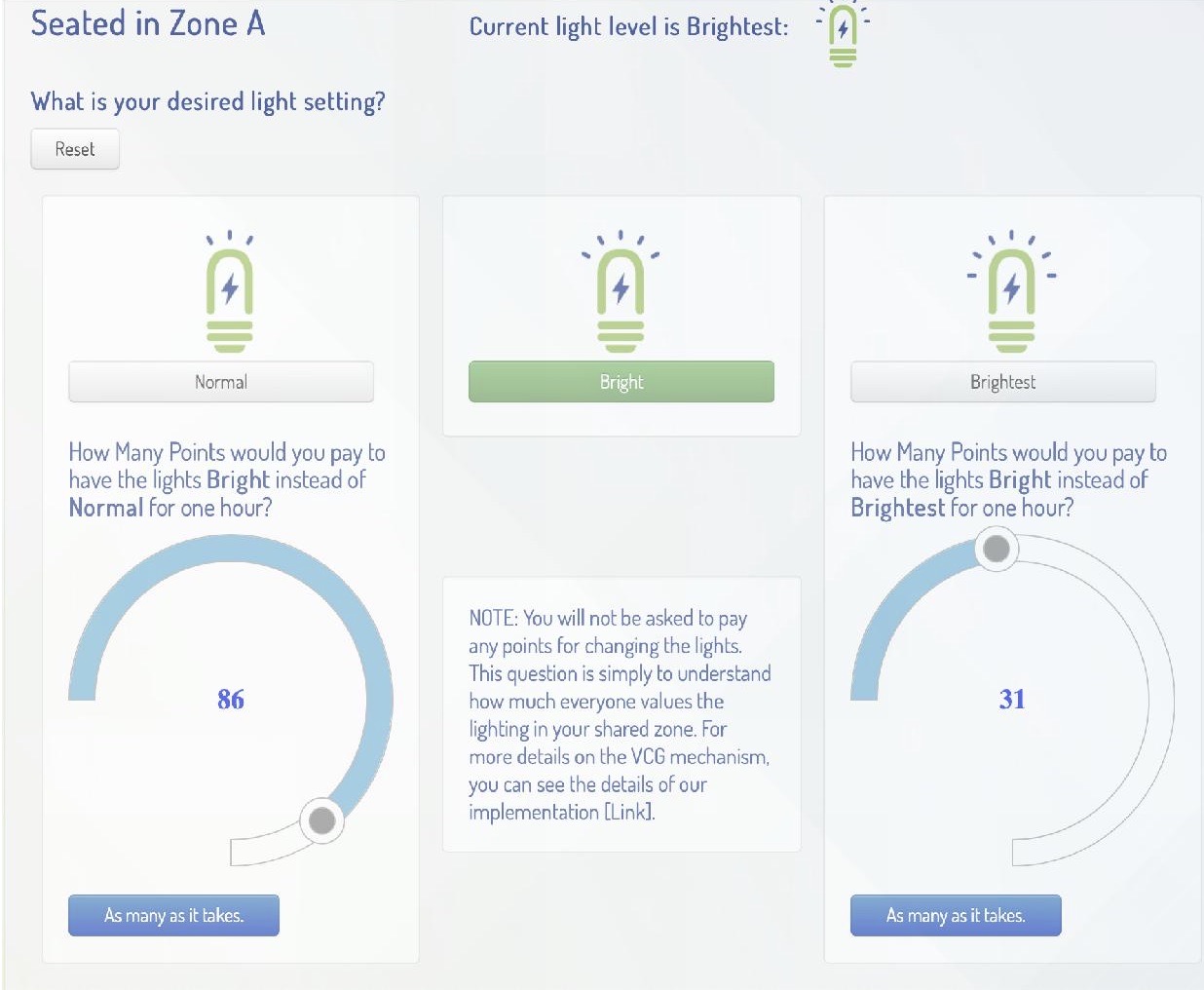}
    \caption[Voting portal for smartSDH]{The voting portal for smartSDH, where users log in and report their preferences on how many points they would pay to change a light setting from the currently implemented setting.}
\end{figure}
To adjust the light conditions, participants first selected their preferred light setting from three available options: `Normal', `Bright', `Very Bright', which corresponds to 33\%, 67\%, and 100\% of the maximum possible illumination. Once participants selected their preferred light setting (e.g.,  `Bright'), the web portal requested follow-up information about the two settings not chosen (e.g., `Normal' and `Very Bright'). Specifically, the web portal asked participants to indicate the extent to which they preferred their chosen light setting over each of the alternative settings (e.g., ``How many points are you willing to pay to have the lights set to BRIGHT instead of NORMAL for one hour?'') from 0 to 100 points (see \ref{fig:portal} for illustrations of the voting procedure). These two measures---preferred light setting and relative preference---contributed to the VCG mechanism's chosen outcome. Every time a user logged on, logged off, or changed their reported preference, smartSDH calculated \(f(\hat \lambda)\) based on the reported values of all users who were currently logged on, and rewards each user \(i\)'s account with points \(p_i(\hat \lambda)\), according to the calculations in Section 3. Users were able to modify their reported preferences at any time, however they were required to be present in the office to vote. Office presence was enforced via the browser's geolocation data.

In addition to the points earned from the VCG mechanism, we also rewarded users for completing a repeatable survey about their experiences with smartSDH. This repeatable survey included Likert-scale questions about the incentives provided, the design of the web portal, the level of comfort participants experienced, their satisfaction with the current light setting, the level of awareness about energy-saving actions they could take, and their productivity in the office.

\subsection{Rewards}

The points were converted to values by the users in two ways throughout the study. First, whenever the total points earned by all participants exceeded a threshold, we held a lottery for gift cards. The probability of one winning the lottery was proportional to the number of points in one's account, and multiple gift cards were given each time the lottery threshold was reached. This was the individual incentive. Second, their points built toward a communal incentive. Whenever the total points earned by all participants exceeded the communal incentive threshold, we provided a catered lunch to participants.

\subsection{Apparatus}

The users interacted with a web portal where they submitted their preferences for the overhead lighting in the office. The Building Automation and Control (BACnet) protocol was used for communications between the web portal interface and the light actuators in the space. We used BACNet because it is the prominent protocol for HVAC applications. 

\subsection{Measures}

\begin{enumerate}

    \item \textbf{Light setting preference.} Participants selected one of three available settings: (`Normal', `Bright', `Very Bright'), which corresponds to 33\%, 67\%, and 100\% of the maximum possible illumination.
    \item \textbf{Relative preference.} For the two settings not chosen, the web portal asked ``How many points are you willing to pay to have the lights set to PREFERRED instead of ALTERNATIVE for one hour?'' For each of the two alternatives, users provided an integer between 0 and 100. We assumed \(\lambda_{\text{max}} = 100\), where \(\lambda_{\text{max}}\) is as defined in \ref{ass:bounded}. As a quality-of-life feature, we also included a button that allowed users to set their vote to \(\lambda_\text{{max}}\) with one click. 
    \item \textbf{Light level satisfaction.} Two items measured participants' satisfaction with their lighting conditions on a given day, ``I am satisfied with today's lighting conditions,'' and ``Today's lighting conditions were uncomfortable.'' Responses were recorded on a five-point Likert-type scale from 1 (strongly disagree) to 5 (strongly agree) ($\alpha$ = .91). 
    \item \textbf{Incentive satisfaction.} Two items measured participants' satisfaction with the incentives provided on a given day, ``I am happy with the current incentives provided,'' and ``The current incentives are not satisfactory.'' Responses were recorded on a five-point Likert-type scale from 1 (strongly disagree) to 5 (strongly agree) ($\alpha$ = .86).
    \item \textbf{User interface satisfaction.} Two items measured participants' satisfaction with the web portal used to manipulate the smart lights on a given day, ``I am satisfied with the current web interface,'' and ``The web portal leaves much to be desired.'' Responses were recorded on a five-point Likert-type scale from 1 (strongly disagree) to 5 (strongly agree) ($\alpha$ = .78).
    \item \textbf{Energy consumption (\% savings/time).} We measured energy consumption in terms of percentage savings. That is, we calculated the implemented lighting over a baseline of 100\% lighting. This is depicted in Fig. \ref{fig:zonea}.
    \item \textbf{Humidity (\%)}. Atmospheric sensors in the office space measured relative humidity in percentage water vapor -- where 100\% corresponds to fully saturated air at dewpoint.  Atmospheric sensors placed in the room.
    \item \textbf{Temperature (deg F)}. The above-mentioned sensors measured the temperature. Another atmospheric condition captured by the aforementioned sensors.
    \item \textbf{Pressure (Hg)}. The sensors measured barometric pressure in terms of units of mercury (Hg). More precisely, a unit of Hg denotes the pressure exerted by a column of mercury 1 inch (25.4 mm) in height at the standard acceleration of gravity. Another atmospheric condition captured.
    \item \textbf{Solar radiation (\(W/m^2\))}. Also known as solar irradiance, and it is measured in terms of power per unit area (Watts per square meter, in this case).

\end{enumerate}

\section{Results}
\label{sec:results}
\begin{center}
\begin{figure}[h!]
    \includegraphics[width=0.52\textwidth, height=0.35\textwidth]{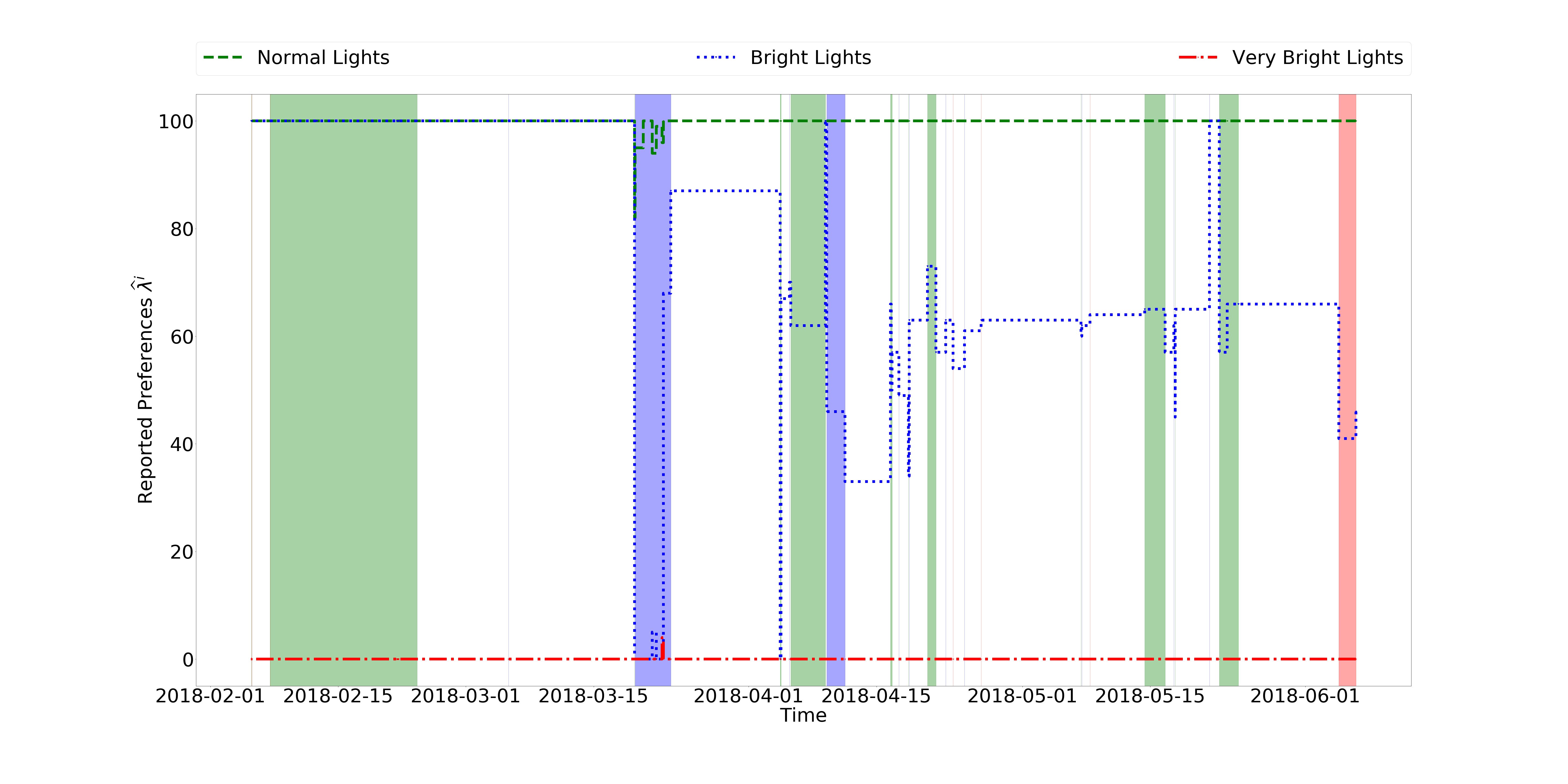}
    \vspace{-8mm}
    \caption[Voting history for a player across all lighting options for the entire gaming period.]{Voting history for a player across all lighting options for the entire gaming period. The unshaded rectangular regions in the background represent the player's absence. The green region indicates normal setting, purple indicates bright and red indicates very bright setting.}
    \label{fig:vcg_choices}
\end{figure}    
\end{center}
\vspace{-5mm}
Fig. \ref{fig:vcg_choices} shows the voting behavior of a user throughout the five-month period of the game. The y-axis, Reported Preferences, is the number of points the user is willing to pay to change the lights when it is at that setting. For instance, in this case, the user is willing to pay nothing to change the ``Very Bright Lights'' setting because it is their most preferred setting and would prefer that it remains at that setting. The shaded regions in the background show the implemented light setting during that time; blank regions indicate that the user was not logged in at the time or actively voting. In this instance, the user initially signals that they do not favor the dimmest light settings and medium light settings and the VCG mechanism implements the dimmest light. As time evolves, the user realizes they are more willing to accept the medium light setting since compromise has proven to output a brighter light overall, as seen in the short instances of blue and, eventually, red shading in the background.

\subsection{Effect of smartSDH on light level and incentive satisfaction across time (RQ1)}
 
 Satisfaction with the light levels across time while controlling for user interface satisfaction was analyzed using a 3 (T1 vs. T2. vs. T3) within-subjects ANOVA. We observed a significant effect of time on light level satisfaction after controlling for user interface satisfaction, \(F(2, 151)\) = 4.21, \(p = .017\), \(d = .217\). Participants reported greater satisfaction with light levels at T1 \((M = 3.44\), \(SD = 1.03)\) than T2 \((M = 2.94\), \(SD = 1.11\), \(p = .039)\) and T3 \((M = 2.96\), \(SD = 1.31\), \(p = .018)\). These results suggest that after an initial period of satisfaction with the light levels, participants became less satisfied with light levels (Fig. \ref{fig:time}). 
 
 \begin{center}
     \begin{figure}[h]
    \centering
    \includegraphics[scale=0.5]{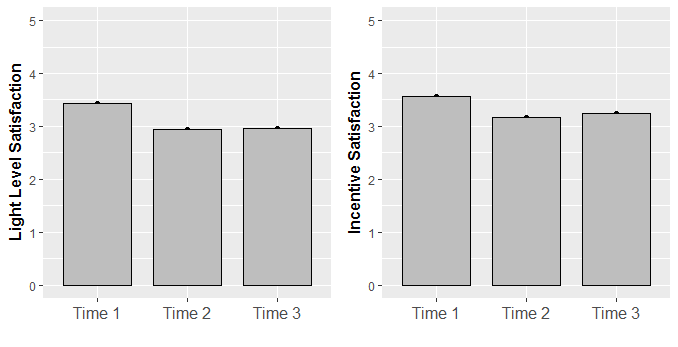}
    \caption{Effect of smartLight on light level and incentive satisfaction across time (N=157).}
    \label{fig:time}
\end{figure}
 \end{center}
\vspace{-5mm} 
Satisfaction with incentives across time while controlling for user interface satisfaction also was analyzed using a 3 (T1 vs. T2. vs. T3) within-subjects ANOVA. We observed a significant effect of time on incentive satisfaction after controlling for user interface satisfaction, \(F(2, 151) = 3.70\), \(p = .027\), \(d = .171\). Participants reported greater satisfaction with incentives at T1 \((M = 3.57\), \(SD = .82)\) than T2 \((M = 3.16\), \(SD = 1.11\), \(p = .001)\) and T3 \((M = 3.24\), \(SD = 1.37\), \(p = .001)\). Despite the theoretical properties of VCG mechanisms, the experimental data suggest the VCG mechanism does not have an appreciable effect on light level and incentive satisfaction of users, shown in (Fig. \ref{fig:time}). 

\vspace{-3mm}
\subsection{Energy saving during smartSDH use over 5-month period (RQ2)}

To evaluate the effect of the VCG-operated smartSDH on energy consumption, we compared the intensity of the overhead lights in the office space over the 5-month period of smartSDH use to the intensity of the overhead light during normal operation (i.e., 100\% of the maximum possible illumination). Our results indicate that employing smartSDH in the office from 9AM to 5PM reduced energy consumption by 35.22\% over the 5 month study period. 

\begin{figure}[h!]
    \centering
    \includegraphics[scale=0.18]{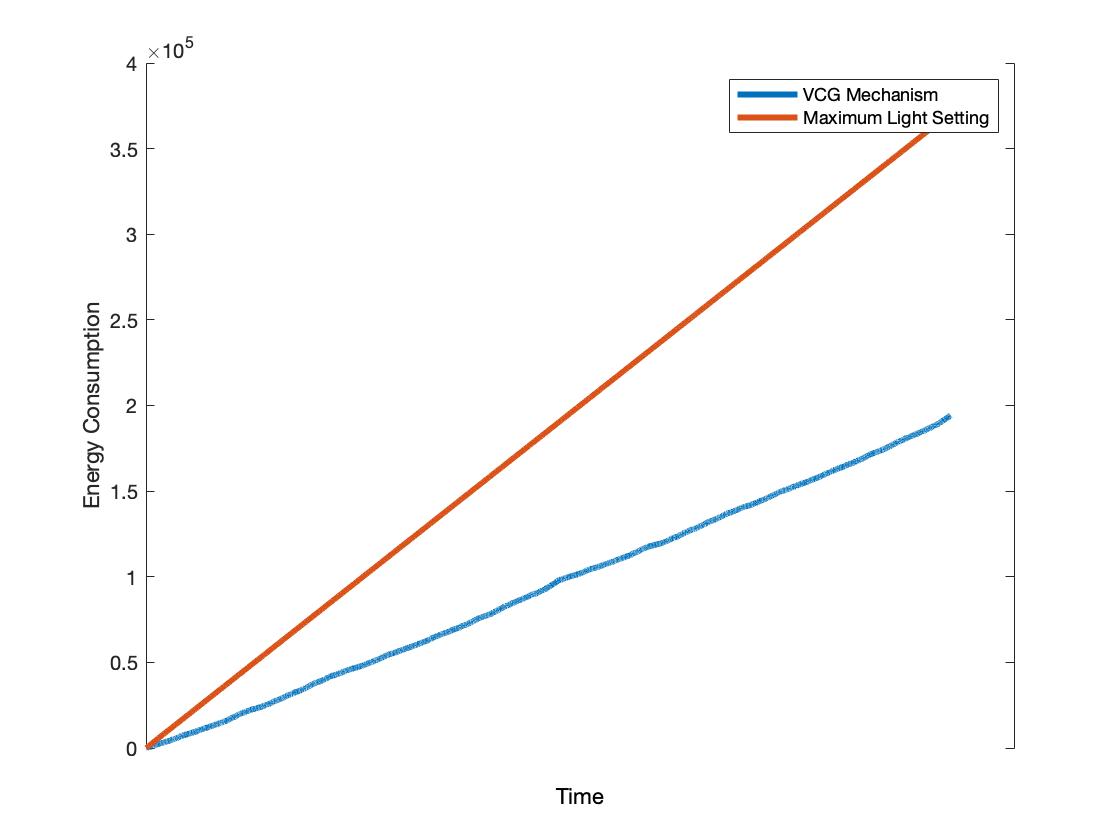}
    \caption{Light energy consumption across time in units used by VCG (100 for brightest, 67 for bright and 33 for normal) compared with maximum light setting across the time period.}
    \label{fig:zonea}
\end{figure}

\subsection{Relation between reported light setting preference and atmospheric conditions (RQ3)}

\begin{table}[htpb]
\caption{Descriptive statistics and correlations for atmospheric variables on reported light level preferences}
\begin{center}
\begin{tabular}{|c|c|c|c|c|}
\hline
\toprule
 & \textit{Mean} & \textit{SD} & \textit{r} & \textit{p} \\
 \midrule
Light setting preference & 61.81         & 27.84       & 1.00       &            \\ 
Temperature (deg F)      & 55.00         & 5.17        & .017       & .005       \\ 
Pressure (Hg)            & 29.62         & 0.11        & .06        & .292       \\
Solar radiation (\(W/m^2\))   & 249.21        & 370.72      & .27        & .001       \\ 
Humidity (\%)            & 79.93         & 11.05       & -.18       & .003       \\
\bottomrule
\hline
\end{tabular}    
\end{center}

\label{tab:atmos}
\bigskip

\small\textit{Note}. \textit{r} represents the zero-order correlation with light level preference. \(p\) represents the corresponding \(p\)-value for each correlation. \(N = 276.\)
\end{table}

Table \ref{tab:atmos} presents means, standard deviations, and bivariate correlations for atmospheric conditions and reported light setting preference. We observed a significant relationship between participants' reported light setting preference and atmospheric data taken in real time mainly solar radiation and humidity. Pressure did not significantly correlate with reported light setting preference, p = .292. We take these results to suggest that mechanisms for building control may require users to constantly report their preferences, as it would be difficult to build estimators of these preferences from externally observable factors.

\section{Discussion}
\label{sec:discuss}

In this study, we implemented a modified VCG mechanism, smartSDH, to determine the brightness of overhead lights in a shared office space. The goal of smartSDH was to determine the light setting that maximized the sum of everyone's utility fairly. To do so, smartSDH would issue incentives as needed to promote truthful reporting of preferences, as well as ensure all users were better off with smartSDH than a nominal light setting along with no rewards. To this end, VCG mechanisms were an appropriate choice for satisfying many of our desiderata.

Our first research question dealt with the influence of smartSDH on light and incentive satisfaction. 
In theory, smartSDH should improve the satisfaction of participants with the lighting. Regarding the first hypothesis on measuring the influence of time on light and incentive satisfaction, quantitative results shown in section~\ref{sec:results} suggest that there is no appreciable increase in satisfaction. We see that in T1, there was greater satisfaction, on average, across all participants than compared to time periods T2 and T3. Additionally, participants were significantly more satisfied with incentives during T1; however, this satisfaction deteriorated in periods T2 and T3. This could be due to a variety of reasons. First, a few survey respondents stated that they were ``too busy'' to regularly report their preferences, which implies that the user interface for the portal was not as convenient as desired. So as time went on, people felt the task more burdensome resulting in a more sporadic voting pattern as time went on. 

Secondly, there is evidence in the literature about ``technology burnouts'' (see~\cite{armel2013disaggregation} and the references therein). These references discuss the issues with technology adoption, highlighting the fact that with the right system in place, it is possible for behavioral approaches such as the one in smartSDH to be much more effective in reducing energy consumption, which can improve the result for our second hypothesis on energy saving. This suggests that there was initially a novelty effect for our users, who were more active in the first time period. In the second and third time periods, users were generally less engaged and motivated to continue voting or simply put, the users experienced a technology burnout. In general, if the user population was more aware of appliance information (in our case, about the lights' energy consumption), it would facilitate greater energy savings. For future work, this suggests more testing on how the user interface is perceived to obtain better insights for design decisions as well as increasing users' awareness on the mechanics of the sensors. If the technology burnout can be contained or even delayed, this methodology can be adopted for a longer period of time with adequate user satisfaction and savings on energy consumption. 

Our second research question dealt with the influence of smartSDH on energy consumption. 
We found that there is a significant reduction in energy usage. This is noteworthy because though the participants generally felt inconvenienced and dissatisfied with the implemented mechanism, their behaviors ended up effectively reducing overall energy usage. For future directions of the work, it might be useful to note that this framework readily allows one to incorporate incentives for energy-efficient behavior. For example, if a building costs $c_i \in \mathbb{R}$ dollars to operate at setting $x_i \in \mathcal{X}$ for one hour, we can introduce another participant of type $\lambda^0 = (c_1,c_2,\dots,c_m)$. This term now shows up in each participant's utility function through the VCG mechanism, and if everyone gains more than $c_i - c_j$ in utility from choosing $i$ over $j$, then the mechanism will choose $i$ over $j$. Each participant can then internalize the incentives, directly experience its benefits, and, as a group, decide whether or not the difference in cost warrants a change in settings. This would allow users to internalize the externality -- in other words, to shift the external cost to their own internally accrued cost. However, for our experiment itself, we did \emph{not} incorporate any energy-efficiency incentives; our primary goal was to find the light settings that were the most preferred by all the occupants in the office space. Moreover, technology adoption rates are generally higher for software solutions such as smart meters (analogous to the sensors used in smartSDH) since the hardware is generally installed by an external party such as the utility company. Since this has virtually no cost to the user, along with zero personal installation effort, there is much to be said about integrating smart software solutions to reduce energy consumption in an aggregate energy source.

Our last research question dealt with the relationship between the preferences of users and atmospheric conditions. 
Our results show that there is a significant relationship between participants' votes and factors such as solar radiation and humidity. However, there is not much qualitative information that we could obtain about why participants felt that these factors influenced them to vote in a particular manner. These qualitative aspects of the user's voting can be asked in a survey at the end of the experiment, which is future work when conducting such experiments. The data from this experiment suggests that estimating these preferences from environmental factors is not a viable avenue, because it indicates that mechanisms for smart building control may require occupants to consistently report their preferences. This is because these preferences cannot be accurately predicted from externally observable factors. In addition, the free-form qualitative survey also suggests that people experienced no appreciable increase in satisfaction as a result of the mechanism. We explore a few potential reasons for this, all of which are interesting avenues for future study.
 
We note that, due to limitations of this study, our geolocation methods required users to regularly sign on from their browser; otherwise, users could vote remotely or leave their computer on to earn points illicitly. 
We hope that future studies can explore methods to reduce the intrusiveness of these mechanisms.

Another possible reason for satisfaction levels not being significant might have been because $\lambda_{max}$, as defined in Assumption~\ref{ass:bounded}, was set too low. This again was due to the limitations of our study. Our value for $\lambda_{max}$ was determined by our experimental budget. That is, we set $h_i = n \lambda_{max}$ as a base rate for participation, and needed to bound how much money we would give out in rewards over the duration of the experiment. For smartSDH, our budget was over \$100 per participant. In practical applications, it may be beneficial to relax the requirement that all users are better off due to the mechanism. Under this relaxation, some users may wind up paying money into the mechanism, but it would still implement the socially optimal outcome. This ties back in with the idea of ``internalizing the externality'' that may be caused due to differences in preferences of lighting. To implement this, there must be a way to force participation against a user's will. 
We note that this is not likely to be viable in an academic research experiment. Alongside the above considerations, it is important to note that since VCG is generally employed in auctions, the users involved are assumed to be competitive. According to assumption \ref{ass:info_structure}, it is assumed that one a user knows their type and others are unaware of it. However, due to the nature of the experiment or the fact that the web interface shows the users in the office environment at any instant, the users might cooperate with one another in determining their types, which might once again have an impact on the optimality of the VCG mechanism-derived outcome.

Finally, we note that many qualitative aspects of our study suggest a status quo effect. Some participants told us during the catered events that they were happy with smartSDH because the lights were typically too bright. In contrast, some other participants insisted that no one would want lights so dim. Anecdotally, it seemed some participants were dissatisfied because they did not believe that the mechanism was implementing the social optimum. It is interesting to note how it is more socially acceptable to walk into a full office and turn up the lights than it is to walk into a full office and turn down the lights; these social contexts likely had a factor in the experience of participants of this study. We think examining the effect of the status quo when new IoT technologies are deployed in these settings is a very interesting direction for future research.

\subsection{Concluding Remarks}

In summary, the contributions of this paper are as follows:

\begin{itemize}

\item The implementation of modified VCG mechanism to determine the brightness of the overhead lights for 27 participants over a 5-month period in their actual office space, giving out \$2,900 in rewards. We emphasize that it is challenging to implement the VCG mechanism outside of typical auction settings (see \ref{sec:mec_des}) and most of the work in the realm of VCG based mechanisms have been strictly theoretical and simulation based.
\item Studying user satisfaction with incentives and light level and the impact of the mechanism on energy savings and consumption.

\item Determining correlations between light setting preferences and atmospheric conditions. Although some factors had significant correlations with the reported preferences, the predictive value was generally very poor. This implies that mechanisms for building control may require users to constantly report their preferences, as it would be difficult to build estimators of these preferences from externally observable factors.

\item Outlining some barriers to the implementation of VCG mechanisms in IoT settings, and potential reasons our study did not achieve the expected gains in satisfaction.

\end{itemize}

Despite the theoretical properties of VCG mechanisms, the experimental data suggests that the VCG mechanism did not have an appreciable effect on the satisfaction or awareness of users in any of these user-perception categories. Our data suggests that rational agent models may require some modifications to capture how humans typically will interact with an IoT technology in the background of their work life. In particular, the preferences reported show some temporal variation that may be due to users learning their own preferences, which is in contrast with the typical rational agent model. Furthermore, there is a need to design the user experience to be as minimally intrusive as possible, and potentially relax the requirement that all users are better off with the existence of the mechanism.

\section*{Acknowledgements}
Ioannis C. Konstantakopoulos was a Scholar of the Alexander S. Onassis Public Benefit Foundation. 
The authors would like to thank Chris Hsu, the applications programmer at CREST Laboratory, UC Berkeley, who developed and deployed the web portal application as well as the social game data pipeline architecture. We would like to thank Andrea Bajcsy for many helpful discussions on best practices for experiments with human behavior. We would also like to thank Claire Tomlin, without whom this experiment may never have been deployed. Additionally, we want to thank UC Berkeley Student Technology Fund for the Award of Student Technology Fund Initiative -- Large-Scale Project (2018). Finally, we would like to thank the participants of smartSDH.

{\footnotesize
\bibliographystyle{ieeetr}
\bibliography{bare_jrnlr.bib}
}

\begin{IEEEbiography}[{\includegraphics[width=1.11in,height=1.25in,clip]{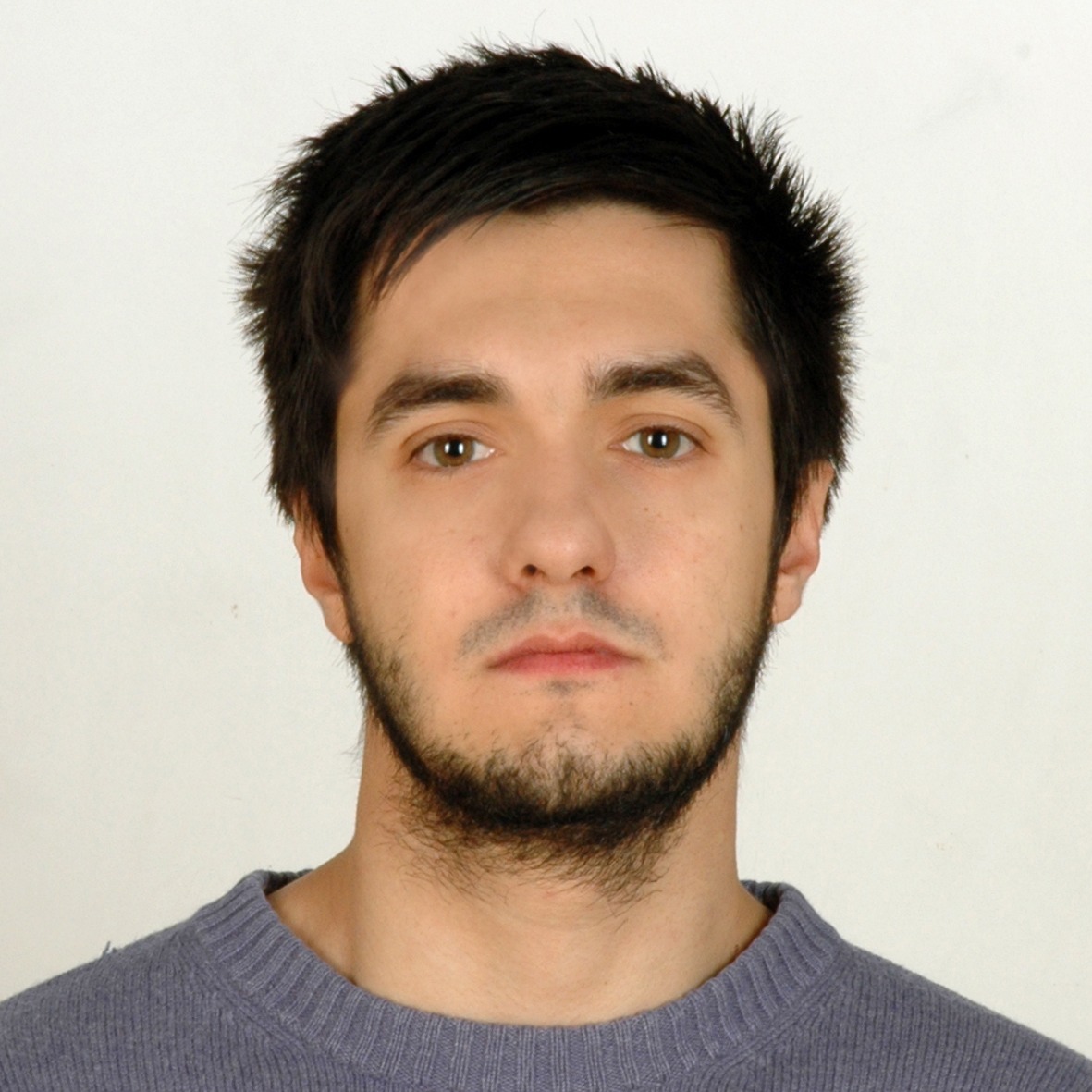}}]{Ioannis C. Konstantakopoulos}
Ioannis C. Konstantakopoulos received the Diploma (Hons.) in electrical and computer engineering from the University of Patras, Patras, Greece, in 2012 and the M.S./PhD degree in electrical engineering and computer sciences from the University of California, Berkeley, CA, USA, in 2014 and 2018 respectively. His current research interests include deep learning, statistical learning, machine learning, algorithmic game theory, and optimization. Ioannis C. Konstantakopoulos was a scholar of the Alexander S. Onassis Public Benefit Foundation providing financial support for outstanding Greek Doctoral students studying at US Institutions. Currently, he is an Applied Scientist at Alexa Artificial Intelligence team at Amazon.com, Inc.
\end{IEEEbiography}
\vspace{-10mm}
\begin{IEEEbiography}[{\includegraphics[width=1.12in,height=1.25in,clip,keepaspectratio]{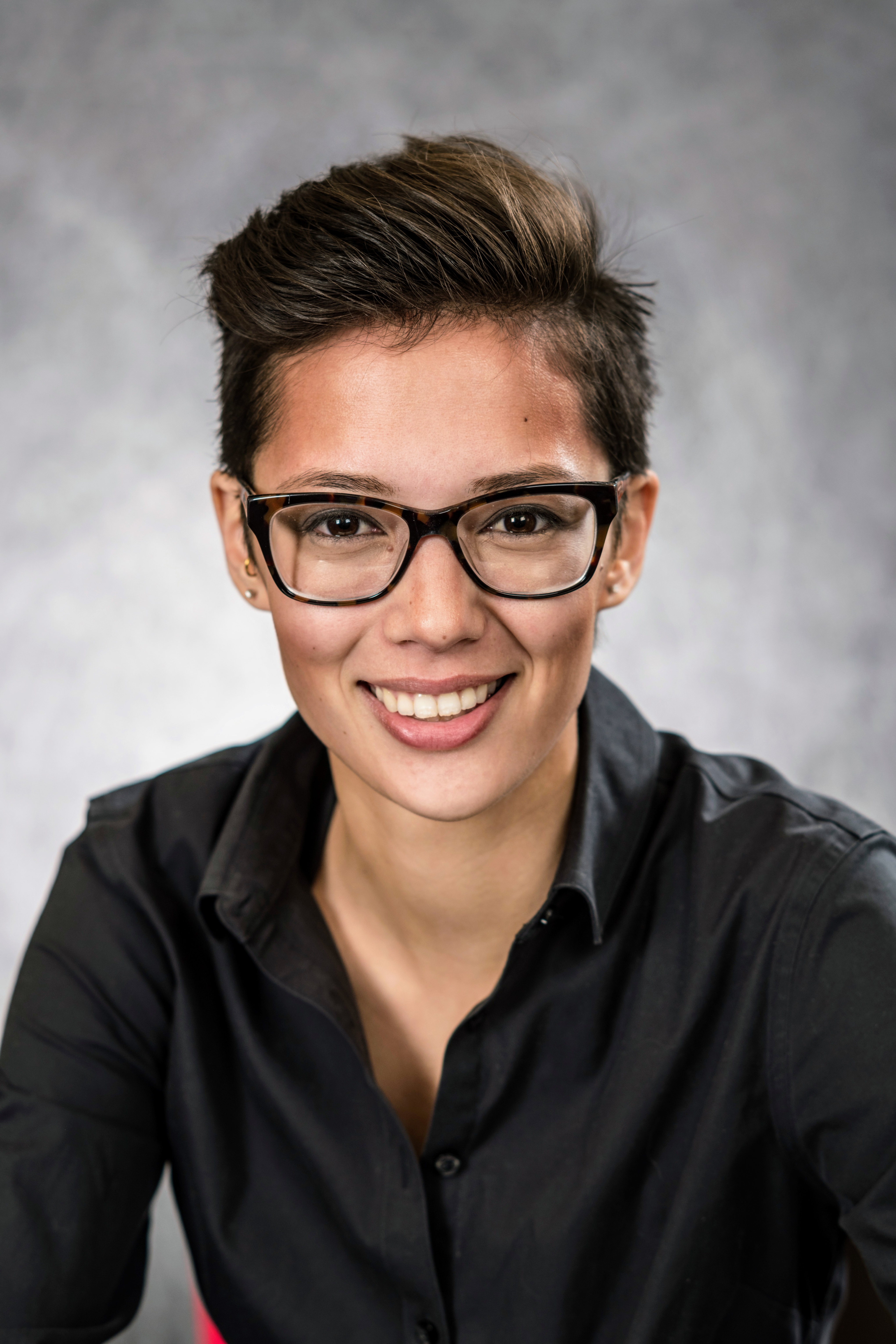}}]{Kristy A. Hamilton}
Kristy A. Hamilton is an Assistant Professor of Digital Media in the Department of Communication at University of California, Santa Barbara. Dr. Hamilton uses experimental methods from cognitive and social psychology to understand the inherent qualities or liabilities of human memory and cognition in a digital environment. She received her PhD from the Institute of Communication Research at the University of Illinois at Urbana-Champaign in 2020. She has published in journals such as Computers in Human Behavior, \emph{Journal of Applied Research in Memory and Cognition, Applied Cognitive Psychology}, and \emph{New Media \& Society}.
\end{IEEEbiography}


\begin{IEEEbiography}[{\includegraphics[width=1.12in,height=1.25in,clip,keepaspectratio]{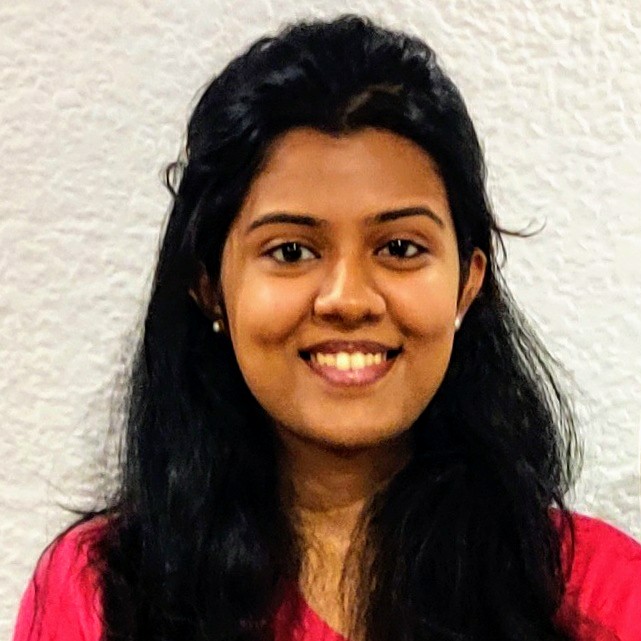}}]{Yashaswini Murthy}
Yashaswini Murthy is currently pursuing a PhD in Electrical Engineering at the University of Illinois, Urbana-Champaign. Her interests lie in the domain of control theory. She graduated with a Bachelor's and Master's of Technology in Mechanical Engineering with specialisation in Computer Aided Design and Automation and a Minor in Systems and Control Engineering from the Indian Institute of Technology, Bombay (IITB) in 2019.
\end{IEEEbiography}

\begin{IEEEbiography}[{\includegraphics[width=1.12in,height=1.25in,clip,keepaspectratio]{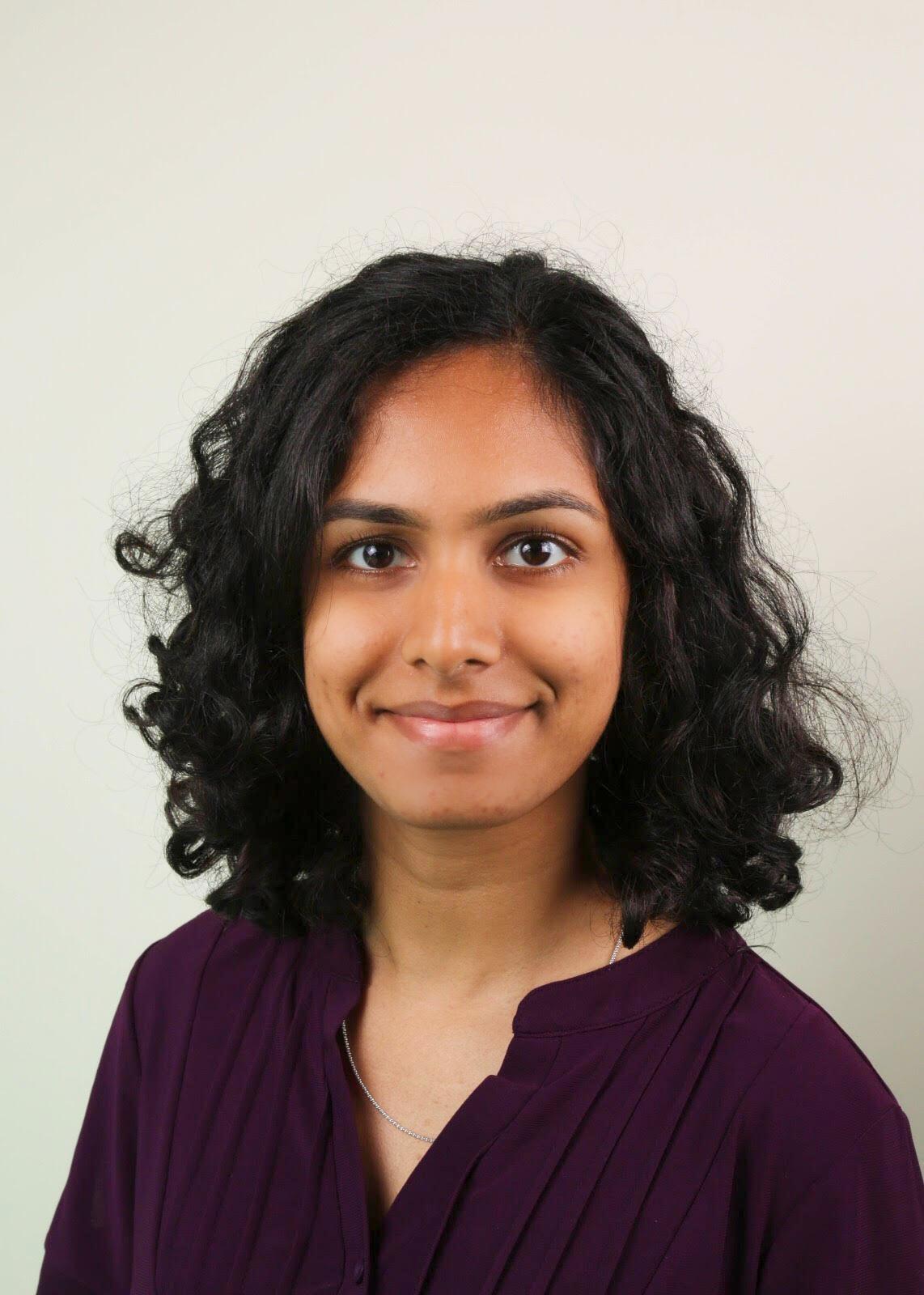}}]{Tanya Veeravalli}
Tanya Veeravalli is a graduate student at the University of Illinois, Urbana-Champaign, pursuing a Ph.D. in Electrical and Computer Engineering. She graduated from the University of California, Berkeley, receiving a B.A. (2019) with majors in Computer Science and Economics. Her interests lie in the areas of learning in dynamical systems and optimization.
\end{IEEEbiography}

\begin{IEEEbiography}[{\includegraphics[width=1.12in,height=1.25in,clip,keepaspectratio]{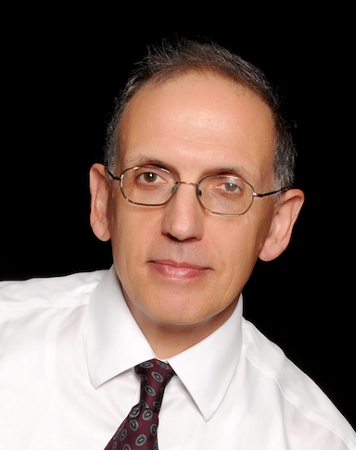}}]{Costas Spanos}
COSTAS J. SPANOS (Fellow, 2000) received the EE Diploma from the National Technical University of Athens, Greece, and the M.S. and Ph.D. degrees in ECE from Carnegie Mellon University. In 1988 he joined the department of Electrical Engineering and Computer Sciences (EECS) at the University of California, Berkeley, where he is now the Andrew S. Grove Distinguished Professor and the Director of the Center for Information Technology Research in the Interest of Society and the Banatao institute (CITRIS). He is also the Founding Director and CEO of the Berkeley Education Alliance for Research in Singapore (BEARS), and the Lead Investigator of a large research program on smart buildings based in California and Singapore. Prior to that, he has been the Chair of EECS at UC Berkeley, the Associate Dean for Research in the College of Engineering at UC Berkeley, and the Director of the UC Berkeley Microfabrication Laboratory. His research focuses on Sensing, Data Analytics, Modeling and Machine Learning, with broad applications in semiconductor technologies, and cyber-physical systems.
\end{IEEEbiography}

\begin{IEEEbiography}[{\includegraphics[width=1.12in,height=1.25in,clip,keepaspectratio]{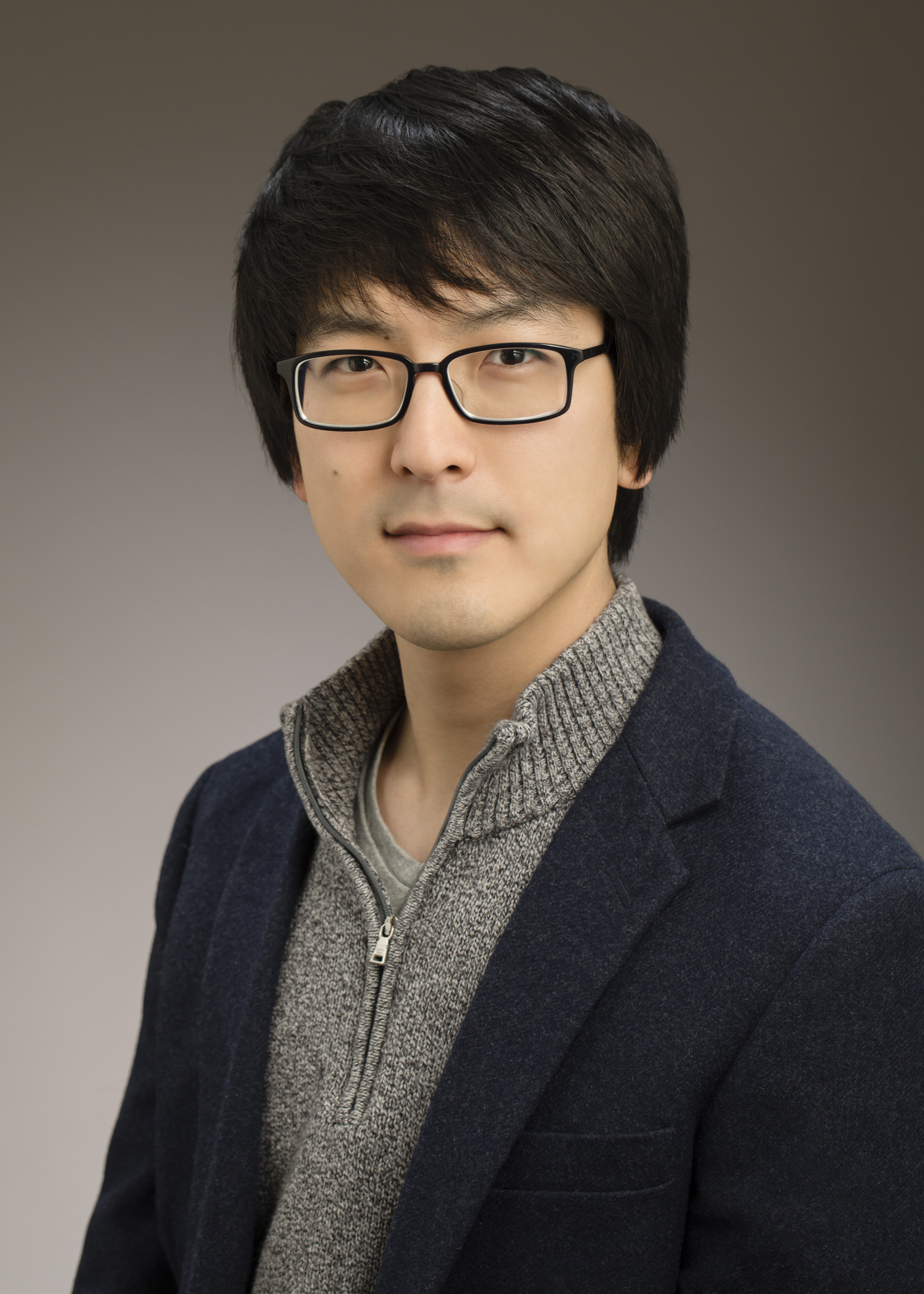}}]{Roy Dong}
Roy Dong is a Research Assistant Professor in the Electrical and Computer Engineering department at the University of Illinois at Urbana-Champaign. He received a BS Honors in Computer Engineering and a BS Honors in Economics from Michigan State University in 2010. He received a PhD in Electrical Engineering and Computer Sciences at the University of California, Berkeley in 2017, where he was funded in part by the NSF Graduate Research Fellowship. From 2017 to 2018, he was a postdoctoral researcher in the Berkeley Energy \& Climate Institute and a visiting lecturer in the Industrial Engineering and Operations Research department at UC Berkeley. His research uses tools from control theory, economics, statistics, and optimization to understand the closed-loop effects of machine learning, with applications in cyber-physical systems such as the smart grid, modern transportation networks,
\end{IEEEbiography}




\end{document}